\documentclass[11pt,reqno]{amsart}
\usepackage{amssymb}
\usepackage{amsfonts}
\usepackage{amsmath}
\usepackage{mathrsfs}
\usepackage{stmaryrd}
\usepackage{physics}
\usepackage{braket}
\usepackage{dsfont}
\usepackage{bbold}
\usepackage{graphicx}
\usepackage{relsize}
\usepackage{makecell}
\usepackage[table,xcdraw]{xcolor}
\usepackage{enumerate}
\usepackage[pagebackref, colorlinks = true, linkcolor = blue, urlcolor  = blue, citecolor = red]{hyperref}
\usepackage[margin=1in]{geometry}
\usepackage{enumitem}
\usepackage{mathtools}

\usepackage{graphbox}
\usepackage{comment}
\usepackage{float}
\usepackage[font=scriptsize]{caption}

\newcommand{\Hil}{\mathcal{H}}

\newcommand{\cH}{\mathcal{H}}

\newcommand{\cL}{\mathcal{L}}

\newcommand{\iden}{\mathbb{1}}
\newcommand{\eps}{\varepsilon}
\renewcommand{\epsilon}{\varepsilon}
\renewcommand{\phi}{\varphi}
\newcommand{\overbar}[1]{\mkern 1.5mu\overline{\mkern-1.5mu#1\mkern-1.5mu}\mkern 1.5mu}

\newcommand{\B}[1]{\mathcal{L}({#1})}
\newcommand{\State}[1]{\mathcal{D}({#1})}

\newcommand{\id}{{\rm{id}}}

\newcommand{\supp}{{\rm{supp }}}

\newtheorem{theorem}{Theorem}[section]

\newtheorem*{definition*}{Definition}

\newtheorem{corollary}[theorem]{Corollary}
\newtheorem{lemma}[theorem]{Lemma}
\newtheorem{remark}[theorem]{Remark}

\newtheorem*{conjecture*}{Conjecture}

\newcommand\vertarrowbox[3][6ex]{%
  \begin{array}[t]{@{}c@{}} #2 \\
  \left\uparrow\vcenter{\hrule height #1}\right.\kern-\nulldelimiterspace\\
  \makebox[0pt]{\scriptsize#3}
  \end{array}%
}

\theoremstyle{definition}

\definecolor{darkgreen}{rgb}{0,0.392,0}

\author{Satvik Singh}
\email{satviksingh2@gmail.com}
\address{\parbox{\linewidth}{ Department of Applied Mathematics and Theoretical Physics, \\ University of Cambridge, Cambridge, United Kingdom }}

\author{Nilanjana Datta}
\email{n.datta@damtp.cam.ac.uk}
\address{\parbox{\linewidth}{Department of Applied Mathematics and Theoretical Physics, \\ University of Cambridge, Cambridge, United Kingdom}}

\title{Information transmission under Markovian noise}

\begin{document}

\begin{abstract}
 We consider an open quantum system undergoing Markovian dynamics, the latter being modelled by a discrete-time quantum Markov semigroup $(\Phi^n)_{n \in {\mathbb{N}}}$, resulting from the action of sequential uses of a quantum channel $\Phi$, with $n \in {\mathbb{N}}$ being the discrete time parameter. We find upper and lower bounds on the one-shot $\eps$-error information transmission capacities of $\Phi^n$ for a finite time $n\in \mathbb{N}$ and $\eps \in [0,1)$ in terms of the structure of the peripheral space of the channel $\Phi$. We consider transmission of $(i)$ classical information (both in the unassisted and entanglement-assisted settings); $(ii)$ quantum information and $(iii)$ private classical information.
\end{abstract}

\maketitle

\section{Introduction}
The time evolution of an {\em{open quantum system,}} i.e.~one which interacts with its surroundings, can be considered to be Markovian when the interactions are assumed to be weak. In this scenario, the dynamics of the system can be modelled by a quantum Markov semigroup (QMS). More precisely, let $A$ denote an open quantum system with the associated finite-dimensional Hilbert space $\Hil_A$ and the space $\cL(\cH_A)$ of linear operators acting on $\cH_A$. Then, in the discrete-time case, the QMS is given by $(\Phi^n)_{n \in {\mathbb{N}}}$, where $\Phi :\cL(\cH_A) \to \cL(\cH_A)$ is a quantum channel (i.e.~a linear completely positive trace-preserving map), $\Phi^n = \Phi \circ \Phi \cdots \circ \Phi$ denotes the $n$-fold composition of $\Phi$ with itself, and $n \in {\mathbb{N}}$ plays the role of the discrete time parameter. The nomenclature is justified by the fact that $(\Phi^n )_{n\in \mathbb{N}}$ satisfies the semigroup property: $\forall n,m \in {\mathbb{N}}: \Phi^{n+m}= \Phi^n \circ \Phi^m$.

In this paper, we consider a quantum system/memory $A$ with a Markovian noise model $(\Phi^n )_{n\in \mathbb{N}}$ as defined above. Our task is to store as much information in the memory as possible, in such a way that it can be reliably recovered (with some error $\epsilon\in [0,1)$) after the memory is left to evolve for some time $n\in \mathbb{N}$. Building such a quantum memory that is able to store information for a long time is crucial in order to build a reliable quantum computer, and consequently, this task has been studied from different perspectives \cite{Tehral2015memory, Brown2016memory}. In this paper, we adopt a Shannon-theoretic viewpoint, where we want to analyze the maximum amount of information that can be stored in memory without placing any physical or computational restrictions on the encoding and decoding operations. Put differently, we are interested in characterizing the {\em{one-shot}} $\epsilon-$error information-transmission capacities of $\Phi^n$ for a given error $\epsilon\in [0,1)$ and time $n\in \mathbb{N}$. This problem was recently studied in \cite{singh2024zero} in the asymptotic time limit $n\to \infty$ for error $\epsilon=0$ (see also \cite[Section V]{guan2016zero}), and in \cite{fawzi2024error} in the asymptotic time limit $n\to \infty$ for error $\epsilon\in[0,1)$.

A quantum channel has different kinds of information-transmission capacities. These depend, for example, on the nature of information being transmitted (classical or quantum), whether there are any auxiliary resources that the sender and receiver might employ in the information-transmission task (e.g.~pre-shared entanglement), and whether the information to be transmitted is private, i.e.~required to be inaccessible to an eavesdropper. In view of these considerations, we study the classical, quantum, entanglement-assisted classical, and private classical capacities in this paper.

Note that in the familiar {\em{asymptotic, memoryless setting}} of quantum Shannon theory, one evaluates capacities of a quantum channel $\Phi$ in the {\em{parallel setting}}. This corresponds to evaluating the optimal rate of information transmission through $\Phi^{\otimes n}$ in the limit $n \to \infty$, under the requirement that the error incurred in the transmission vanishes in this limit. In contrast, here we consider information transmission through $\Phi$ in the {\em{sequential setting}}, that is, through $n$ sequential uses of $\Phi$. In addition, we focus on the more realistic scenario in which we consider using $\Phi$ a finite number of times in succession (i.e.~$n < \infty$) and allow a non-zero probability of error (say, $\eps \in [0,1)$) in transmitting the information through $\Phi^n$.

\medskip

\noindent
{\bf{Layout of the paper:}} In Section~\ref{sec:prelims}
we introduce the relevant mathematical notation and definitions, as well as the information-processing tasks (or protocols) considered in this paper. Our main result is stated in Theorem~\ref{theorem:main}. The theorem provides upper and lower bounds on the one-shot $\eps$-error capacities of $\Phi^n$ for finite $n \in {\mathbb{N}}$ in terms of the structure of the peripheral space of $\Phi$. Some technical lemmas which are used in the proof of Theorem~\ref{theorem:main} are stated and proved in Appendix~\ref{appen:tech}. The corresponding upper and lower bounds on the one-shot $\eps$-error capacities of $\Phi^n$ in the asymptotic time limit ($n\to \infty$) are stated in Corollary~\ref{corollary:main}.

\section{Preliminaries}
\label{sec:prelims}
We denote quantum systems by capital letters $A,B,C$ and the associated (finite-dimensional) Hilbert spaces by $\Hil_A, \Hil_B$ and $\Hil_C$ with dimensions $d_A, d_B$ and $d_C$, respectively. For a joint system $AB$, the associated Hilbert space is $\Hil_A\otimes \Hil_B$. The space of linear operators acting on $\Hil_A$ is denoted by $\B{\Hil_A}$ and the convex set of quantum states or density operators (these are positive semi-definite operators in $\B{\Hil_A}$ with unit trace) is denoted by $\State{\Hil_A}$. For a unit vector $\ket{\psi}\in \Hil_A$, the pure state $\ketbra{\psi}\in \State{\Hil_A}$ is denoted by $\psi$. A quantum channel $\Phi:\B{\Hil_A}\to \B{\Hil_B}$ is a linear, completely positive, and trace preserving map. By Stinespring's dilation theorem, for a quantum channel $\Phi : \B{\Hil_A}\to \B{\Hil_B}$, there exists an isometry $V: \Hil_A \to \Hil_B\otimes \Hil_E$ (called the Stinespring isometry) such that for all $X \in \B{\Hil_A}$,  $\Phi (X) = \Tr_E (VX V^{\dagger})$, where $\Tr_E$ denotes the partial trace operation over the $E$ subsystem. The corresponding complementary channel $\Phi^c : \B{\Hil_A}\to \B{\Hil_E}$ is then defined as $\Phi^c (X) = \Tr_B (VX V^{\dagger}).$ The adjoint $\Phi^*$ of a quantum channel $\Phi : \B{\Hil_A}\to \B{\Hil_B}$ is defined through the following relation: $\Tr(Y \Phi(X))= \Tr (\Phi^*(Y) X)$ for any $X \in \B{\Hil_A}$ and $Y \in \B{\Hil_B}$.
\begin{remark}
    To make the systems on which an operator or a channel acts more explicit, we sometimes denote operators $X\in \B{\Hil_A}$ and linear maps $\Phi:\B{\Hil_A}\to \B{\Hil_B}$ by $X_A$ and $\Phi_{A\to B}$, respectively.
\end{remark}
For a bipartite operator $X_{RA}$ and a linear map $\Phi_{A\to B}$, we use the shorthand $\Phi_{A\to B}(X_{RA})$ to denote $(\id_R \otimes \Phi_{A\to B}) (X_{RA})$, where $\id_R $ is the identity map on $\B{\Hil_R}$. Similarly, $X_R$ and $X_A$ denote the reduced operators on $R$ and $A$, respectively, i.e. $X_R := \Tr_A X_{RA}$ and $X_A := \Tr_R X_{RA}$.

The trace norm of a linear operator $X\in \B{\Hil_A}$ is defined as $\norm{X}_1 := \Tr \sqrt{X^{\dagger}X}$. The diamond norm of a linear map $\Phi: \B{\Hil_A}\to \B{\Hil_B}$ is defined as
\begin{equation}
    \norm{\Phi}_{\diamond} := \sup_{\norm{X}_1\leq 1} \norm{\Phi_{A\to B}(X_{RA})}_1,
\end{equation}
where the supremum if over all $X\in \B{\Hil_R\otimes \Hil_A}$ with $d_R=d_A$ and $\norm{X}_1\leq 1$. We denote the operator norm of $X\in \B{\Hil_A}$ by $\norm{X}_{\infty}$. 

The max-relative entropy between two positive semi-definite operators $\rho, \sigma\in \B{\Hil_A}$ is defined as \cite{Datta2009max}:
\begin{equation}
    D_{\max} (\rho \Vert \sigma) := \log \inf \{ \lambda : \rho \leq \lambda \sigma \},
\end{equation}
where the infimum over an empty set is assumed to be $+\infty$. When $\supp \rho \subseteq \supp \sigma$, 
\begin{equation}\label{eq:Dmax}
D_{\max} (\rho \Vert \sigma) = \log \norm{\sigma^{-1/2} \rho \sigma^{-1/2} }_{\infty}.    
\end{equation}
The max-relative entropy is quasi-convex: for positive semi-definite operators $\{ \rho_i\}_i, \{\sigma_i\}_i \subset \B{\Hil_A}$:
\begin{align}
    D_{\max} \left( \sum_i \rho_i \bigg\| \sum_i \sigma_i \right) \leq \max_i D_{\max}( \rho_i \Vert \sigma_i).
    \label{eq:qconvex}
\end{align}
Moreover, if for each $i$, $\rho_i$ and $\sigma_i$ are supported on a subspace $\Hil_i\subseteq \Hil_A$ such that for $i\neq j$, $\Hil_i \perp \Hil_j$, the inequality above becomes an equality.

The fidelity between two quantum states $\rho, \sigma \in \State{\Hil_A}$ is defined as $F(\rho, \sigma):= \norm{\sqrt{\rho}\sqrt{\sigma}}^2_1$.

The $\epsilon-$hypothesis testing relative entropy between a state $\rho\in \State{\Hil_A}$ and a positive semi-definite operator $\sigma\in \B{\Hil_A}$ with $\epsilon\in [0,1]$ is defined as follows~\cite{WRenner2012hypo}:
\begin{equation}
    D^{\epsilon}_H (\rho \Vert \sigma) := - \log \beta^{\epsilon}_H (\rho \Vert \sigma), 
\end{equation}
where 
\begin{equation}
    \beta^{\epsilon}_H (\rho \Vert \sigma) := \inf_{0\leq \Lambda \leq \iden} \{\Tr \Lambda \sigma \, : \, \Tr \Lambda \rho \geq 1-\epsilon  \}.
\end{equation}

The $\epsilon-$hypothesis testing relative entropy of entanglement of a state $\rho_{AB}$ is defined as 
\begin{equation}\label{eq:EHstate}
    E_{H}^{\epsilon}(A:B)_{\rho} := \inf_{\sigma_{AB}\in \text{SEP}(A:B)} D_H^{\epsilon}(\rho_{AB}\Vert \sigma_{AB}),
\end{equation}
where the infimum is over the set $\text{SEP}(A:B)$ of all separable states in $ \State{\Hil_A\otimes \Hil_B}$.

The $\epsilon-$hypothesis testing relative entropy of entanglement of a channel $\Phi_{A\to B}$ is defined as 
\begin{equation}\label{eq:EHchannel}
    E_{H}^{\epsilon}(\Phi ) := \sup_{\rho_{RA}} E_{H}^{\epsilon}(R:B)_{\omega} = \sup_{\rho_{RA}} \inf_{\sigma_{RB}\in \operatorname{SEP}(R:B)} D_H^{\epsilon} (\Phi_{A\to B}(\rho_{RA}) \Vert \sigma_{RB}),
\end{equation}
where the supremum is over all states $\rho_{RA}$ and $\omega_{RB} = \Phi_{A\to B} (\rho_{RA})$. It can be shown that the supremum here can be restricted to pure states $\psi_{RA}$ with $d_R=d_A$.

\subsection{Classical communication}\label{subsec:cc}
An $(\mathscr{M},\epsilon)$ classical communication protocol with $\mathscr{M}\in \mathbb{N}$ and $\epsilon\in [0,1)$ for a channel $\Phi_{A\to B}$ consists of the following:
\begin{itemize}
    \item Encoding states $\rho^m_{A}$ that Alice uses to encode a message $m\in [\mathscr{M}] := \{1,2,\ldots ,\mathscr{M} \}$.
    \item Decoding POVM $\{\Lambda^m_B\}_{m \in [\mathscr{M}]}$ that Bob uses to decode the message,
\end{itemize}
such that for each message $m$,
\begin{equation}
    \Tr [\Lambda^m_{B}  (\Phi_{A\to B}(\rho^m_A)  ] \geq 1 -\epsilon.
\end{equation}
The one-shot $\epsilon-$error classical capacity of $\Phi$ is defined as
\begin{equation}
    C_{\epsilon}(\Phi):= \sup \{ \log \mathscr{M}: \exists (\mathscr{M},\bar{\epsilon}) \text{ classical communication protocol for } \Phi \text{ with } \bar{\epsilon}\leq \epsilon  \}.
\end{equation}

\begin{remark}
    In the literature, one-shot communication capacities are usually denoted as $C^{(1)}_{\epsilon}(\Phi)$. However, since all the capacities considered in this paper are one-shot,
    we omit the superscript $(1)$ for notational simplicity.
\end{remark}

\subsection{Private classical communication}
An $(\mathscr{M},\epsilon)$ private classical communication protocol through a channel $\Phi_{A\to B}$ consists of the following:
\begin{itemize}
    \item Encoding states $\rho^m_{A}$ that Alice uses to encode a message $m\in [\mathscr{M}]$.
    \item Decoding POVM $\{\Lambda^m_B\}_{m\in [\mathscr{M}]}$ with an associated channel $\mathcal{D}_{B\to M'}$ defined as \\ $\mathcal{D}(\cdot) = \sum_m \Tr (\Lambda^m_B (\cdot))\ketbra{m}_{M'}$ that Bob uses to decode the message,
\end{itemize}
such that for each message $m$,
\begin{equation}
    F(\ketbra{m}_{M'} \otimes \sigma_E , \mathcal{D}_{B\to M'}\circ \mathcal{V}_{A\to BE} (\rho^m_A)) \geq 1-\epsilon,
\end{equation}
where $\sigma_E$ is some fixed state independent of $m$ and $\mathcal{V}_{A\to BE}(\cdot ) = V (\cdot) V^{\dagger}$, where $V:\Hil_A\to \Hil_B \otimes \Hil_E$ is a Stinespring isometry of $\Phi_{A\to B}$. By using the data processing inequality for the fidelity function, it is easy to show that the above condition implies:
\begin{align}
    \forall m: \quad \Tr [\Lambda^m_{B}  (\Phi_{A\to B}(\rho^m_A)  ] &\geq 1 -\epsilon, \\
    F(\sigma_E, \Phi^c_{A\to E}(\rho^m_A)) &\geq 1-\epsilon,
\end{align}
where $\Phi^c_{A \to E}$ denotes a quantum channel which is complementary to $\Phi_{A \to B}.$
\medskip

The one-shot $\epsilon-$error private classical capacity of $\Phi$ is defined as
\begin{equation}
    C_{\epsilon}^{\text{p}}(\Phi):= \sup \{ \log \mathscr{M}: \exists (\mathscr{M},\bar{\epsilon}) \text{ private classical communication protocol for } \Phi \text{ with } \bar{\epsilon}\leq \epsilon  \}.
\end{equation}

\subsection{Entanglement-assisted classical communication}
An $(\mathscr{M},\epsilon)$ entanglement-assisted classical communication protocol through a channel $\Phi_{A\to B}$ consists of the following:
\begin{itemize}
    \item An entangled state $\psi_{A'B'}$ shared between Alice and Bob,
    \item Encoding channels $\mathcal{E}^m_{A'\to A}$ that Alice uses to encode a message $m\in [\mathscr{M}]$,
    \item Decoding POVM $\{\Lambda^m_{BB'}\}_{m\in [\mathscr{M}]}$ that Bob uses to decode the message,
\end{itemize}
such that for each message $m$, 
\begin{equation}
   \Tr [\Lambda^m_{BB'}  (\Phi_{A\to B} \circ \mathcal{E}^m_{A'\to A} (\psi_{A'B'}) ) ] \geq 1 -\epsilon.
\end{equation}
The one-shot $\epsilon-$error entanglement-assisted classical capacity of $\Phi$ is defined as
\begin{align}
    C_{\epsilon}^{\operatorname{ea}}(\Phi):= \sup \{ \log \mathscr{M}: \exists (\mathscr{M},\bar{\epsilon})& \text{ entanglement-assisted} \\ 
    &\text{classical communication protocol for } \Phi \text{ with } \bar{\epsilon}\leq \epsilon\}.
\end{align}

\subsection{Quantum communication}
A $(d,\epsilon)$ quantum communication protocol $(\mathcal{E}_{A'\to A}, \mathcal{D}_{B\to B'})$ for a channel $\Phi_{A\to B}$ consists of the following $(d=d_{A'}=d_{B'})$:
\begin{itemize}
    \item An encoding channel $\mathcal{E}_{A'\to A}$ that Alice uses to encode quantum information,
    \item A decoding channel $\mathcal{D}_{B\to B'}$ that Bob uses to decode the information,
\end{itemize}
such that for every pure state $\psi_{RA'}$ 
\begin{equation}
   \bra{\psi_{RB'}} \mathcal{D}_{B\to B'}\circ \Phi_{A\to B}\circ \mathcal{E}_{A'\to A}(\psi_{RA'})\ket{\psi_{RB'}} \geq  1-\epsilon.
\end{equation}
The one-shot $\epsilon-$error quantum capacity of $\Phi$ is defined as
\begin{equation}
    Q_{\epsilon}(\Phi):= \sup \{ \log d: \exists (d,\bar{\epsilon}) \text{ quantum communication protocol for } \Phi \text{ with } \bar{\epsilon}\leq \epsilon  \}.
\end{equation}

\subsection{Spectral properties of quantum channels}

Let $\Phi:\B{\Hil_A}\to \B{\Hil_{A}}$ be a quantum channel. Then, $\Phi$ admits a Jordan decomposition \cite[Chapter 6]{Wolf2012Qtour}
\begin{equation}
    \Phi = \sum_{i} \lambda_i \mathcal{P}_i + \mathcal{N}_i \quad \text{with} \quad \mathcal{N}_i \mathcal{P}_i = \mathcal{P}_i \mathcal{N}_i = \mathcal{N}_i \,\,\, \text{and} \,\,\, \mathcal{P}_i \mathcal{P}_j = \delta_{ij}\mathcal{P}_i,
\end{equation}
where the sum runs over the distinct eigenvalues $\lambda_i$ of $\Phi$, $\mathcal{P}_i$ are projectors (i.e.~$\mathcal{P}_i^2= \mathcal{P}_i$) whose rank equals the algebraic multiplicity of $\lambda_i$, and $\mathcal{N}_i$ denote the corresponding nilpotent operators.

All the eigenvalues $\lambda$ of $\Phi$ satisfy $\abs{\lambda}\leq 1$, and $\lambda=1$ is always an eigenvalue. Moreover, all eigenvalues $\lambda$ with $\abs{\lambda}=1$ have equal algebraic and geometric multiplicities, so that $\mathcal{N}_i=0$ for all such eigenvalues. As $n\to \infty$, we expect the image of 
\begin{equation}
    \Phi^n := \underbrace{\Phi \circ \Phi \circ \ldots \circ \Phi}_{n\,  \text{times}}
\end{equation}
to converge to the peripheral space $\chi (\Phi):= \text{span}\{X\in \B{\Hil_A} : \Phi(X)=\lambda X, |\lambda|=1 \}$. We define the asymptotic part of $\Phi$ and the projector onto the peripheral space, respectively, as follows: 
\begin{equation}\label{eq:phiinf-proj}
\Phi_{\infty}:= \sum_{i:\, |\lambda_i|=1}\lambda_i \mathcal{P}_i \quad \text{and} \quad \mathcal{P} = \sum_{i: \, |\lambda_i|=1} \mathcal{P}_i .
\end{equation}
Clearly, $\Phi_{\infty}=\Phi_{\infty}\circ \mathcal{P} = \mathcal{P}\circ \Phi_{\infty}$.
Notably, both $\Phi_{\infty}:\B{\Hil_A}\to \B{\Hil_A}$ and $\mathcal{P}:\B{\Hil_A}\to \B{\Hil_A}$ arise as limit points of the set $(\Phi^n )_{n\in \mathbb{N}}$ \cite[Lemma 3.1]{Szehr2014specconvergence}. Since the set of quantum channels acting on $\Hil_A$ is closed, both $\Phi_{\infty}$ and $\mathcal{P}$ are quantum channels themselves. As $n$ increases, $\norm{\Phi^n - \Phi^n_{\infty}}_{\diamond}$ approaches zero. More precisely, the convergence behavior is like 
\begin{equation}\label{eq:converge}
    \norm{\Phi^n - \Phi^n_{\infty}}_{\diamond} \leq \kappa \mu^n,
\end{equation}
where $\mu = \operatorname{spr}(\Phi-\Phi_{\infty})<1$ is the spectral radius of $\Phi-\Phi_{\infty}$ (i.e. $\mu$ is the largest magnitude of the eigenvalues of $\Phi-\Phi_{\infty}$) and $\kappa$ depends on the spectrum of $\Phi$, on $n$, and on the dimension $d_A$ of $\Hil_A$ \cite{Szehr2014specconvergence}. The dependence of $\kappa$ on $n$ is sub-exponential, which captures the fact that for large $n$, the convergence is governed by an exponential decay as $\mu^n$. It is known that there exists a decomposition $\Hil_A = \Hil_{0} \oplus \bigoplus_{k=1}^K \Hil_{k,1}\otimes \Hil_{k,2}$ and positive definite states $\delta_{k}\in \State{\Hil_{k,2}}$ such that \cite[Chapter 6]{Wolf2012Qtour}: 
    \begin{equation}\label{eq:phasespace}
        \chi (\Phi) = 0 \oplus \bigoplus_{k=1}^K (\B{\Hil_{k,1}}\otimes \delta_k). 
    \end{equation}
Moreover, there exist unitaries $U_k\in \B{\Hil_{k,1}}$ and a permutation $\pi$ which permutes within subsets of $\{1,2,\ldots ,K \}$ for which the corresponding $\Hil_{k,1}$'s have the same dimension, such that for any
    \begin{equation}\label{eq:phaseaction}
        X = 0 \oplus \bigoplus_{k=1}^K x_k \otimes \delta_k, \quad\text{we have}\quad\Phi (X) = 0 \oplus \bigoplus_{k=1}^K U^{\dagger}_k x_{\pi (k)} U_k \otimes \delta_k .
    \end{equation}
Given a channel $\Phi:\B{\Hil_A}\to \B{\Hil_A}$, the structure of its peripheral space $\chi (\Phi)$, i.e., the block dimensions $d_k=\dim \Hil_{k,1}$ and the states $\delta_k$ in Eq.~\eqref{eq:phasespace}, can be efficiently computed (see \cite[Section 4]{fawzi2024error} and references therein).


\section{Main result}
\label{sec:Main}

We can now state and prove our main result.

\begin{theorem}\label{theorem:main}
Let $\Phi:\B{\Hil_A}\to \B{\Hil_A}$ be a quantum channel, $(\Phi^n)_{n\in \mathbb{N}}$ be the associated dQMS, and $\epsilon\in [0,1)$. Then, for all $n\in\mathbb{N}$, the one-shot $\epsilon-$error capacities satisfy:
\begin{align}   Q_{\epsilon}(\Phi^n) &\geq \log (\max_k d_k ), \label{qlo}\\
C_{\epsilon}^{\operatorname{p}} (\Phi^n) &\geq \log (\max_k d_k), \label{cplo}\\\
C_{\epsilon}(\Phi^n) &\geq \log (\sum_k d_k ), \label{clo}\\
C_{\epsilon}^{\operatorname{ea}}(\Phi^n) &\geq \log (\sum_k d^2_k ).\label{cealo}
\end{align}
Moreover, for $n$ large enough, the following converse bounds hold:
\begin{align}
      Q_{\epsilon}(\Phi^n) &\leq \log (\max_k d_k) + \log(\frac{1}{1-\epsilon-\kappa\mu^n}), \label{eq:Qconverse}\\ 
      C_{\epsilon}^{\operatorname{p}}(\Phi^n) &\leq \log (\max_k d_k ) + \log(\frac{1}{1-\epsilon-\kappa\mu^n}) \label{eq:Pconverse} \\
     C_{\epsilon}(\Phi^n) &\leq  \log (\sum_k d_k) + \log(\frac{1}{1-\epsilon-\kappa\mu^n}), \label{eq:Cconverse}\\
     C_{\epsilon}^{\operatorname{ea}}(\Phi^n) &\leq  \log (\sum_k d_k^2 ) + \log(\frac{1}{1-\epsilon-\kappa\mu^n}).\label{eq:Ceaconverse}
\end{align}
Here, $d_k=\dim \Hil_{k,1}$ for $k\in \{1,2,\ldots ,K \}$ are the block dimensions in the decomposition of $\chi (\Phi)$ (see Eq.~\eqref{eq:phasespace}), $\mu, \kappa$ govern the convergence $\norm{\Phi^n-\Phi^n_{\infty}}_{\diamond} \leq \kappa\mu^n \to 0$ as $n\to \infty$ (see Eq.~\eqref{eq:converge}), 
and $n$ is large enough so that $\epsilon + \kappa\mu^n <1$.
\end{theorem}
Before proving the theorem, let us discuss some of its consequences. 

\subsection{Infinite-time capacities}
Firstly, the theorem allows us to easily characterize the infinite-time capacities of any dQMS by taking the limit $n\to \infty$.

\begin{corollary}\label{corollary:main}
For a channel $\Phi:\B{\Hil_A}\to \B{\Hil_A}$ and $\epsilon\in [0,1)$, the following holds true:
    \begin{align}
     \log (\max_k d_k) \leq\lim_{n\to \infty } Q_{\epsilon}(\Phi^n) &\leq  \log (\max_k d_k) + \log (\frac{1}{1-\epsilon}) \label{eq:Qinfty}  \\ 
     \log (\max_k d_k) \leq\lim_{n\to \infty } C^{\operatorname{p}}_{\epsilon}(\Phi^n) &\leq  \log (\max_k d_k) + \log (\frac{1}{1-\epsilon})  \\ 
     \log\left(\sum_k d_k\right)\leq \lim_{n\to \infty } C_{\epsilon}(\Phi^n) &\leq \log\left(\sum_k d_k\right) + \log (\frac{1}{1-\epsilon}) \label{eq:Cinfty} \\
     \log\left(\sum_k d^2_k\right)\leq \lim_{n\to \infty } C^{\operatorname{ea}}_{\epsilon}(\Phi^n) &\leq \log\left(\sum_k d^2_k\right) + \log (\frac{1}{1-\epsilon}),
\end{align}
where $d_k=\dim \Hil_{k,1}$ for $k\in \{1,2,\ldots ,K \}$ are the block dimensions in the decomposition of $\chi (\Phi)$.
\end{corollary}


\begin{remark}
    Eqs.~\eqref{eq:Qinfty} and \eqref{eq:Cinfty} were independently proved in \cite{singh2024zero} (for the $\epsilon=0$ case) and in \cite{fawzi2024error} (for arbitrary $\epsilon\in [0,1)$). The $\epsilon=0$ case of Eq.~\eqref{eq:Cinfty} is also proved in \cite{guan2016zero}. 
\end{remark}


\subsection{Rate of convergence}\label{subsec:convergence}
Given the infinite time capacities of a dQMS as in Corollary~\ref{corollary:main}, it is natural to ask how quickly do the capacities converge to their infinite time values. According to Theorem~\ref{theorem:main}, the rate of convergence crucially depends on the numbers $\kappa, \mu$. From \cite[Corollary 4.4]{Szehr2014specconvergence}, one can show that any dQMS $(\Phi^n)_{n\in \mathbb{N}}$ acting on a $d-$dimensional memory has $\kappa = O(n^{d^2})$ as $d\to \infty$, assuming that $\mu=\operatorname{spr}(\Phi-\Phi_{\infty})$ is bounded away from 1 as $d\to \infty$. Thus, we can say that the noise $\Phi$ `reaches' its infinite-time capacity when $n$ is large enough so that $\mu^n n^{d^2} \leq \delta$ for some threshold $\delta<1$, which happens when $$n \geq C(d^2 \log d +\log 1/\delta)/ (\log 1/\mu).$$ This shows that after  time $O(d^2 \log d)$ (i.e.~exponential in the number of qubits in memory), the noise reaches its infinite-time capacity. Note that in the special case of $\epsilon=0$, it is possible to obtain a slightly stronger convergence estimate, namely that all the zero-error capacities stabilize after time $n\geq d^2$ (see \cite[Theorem I.4]{singh2024zero}):
\begin{equation}
    \forall \bar{n} \geq d^2: \quad Q_0(\Phi^{\bar{n}}) = \lim_{n\to \infty} Q_0(\Phi^n) = \log \max_k d_k.
\end{equation}
We expect that the stated $O(d^2)=O(2^{2m})$ bound (where $m$ is the number of qubits in memory) cannot be improved in general (see, for e.g., the discussion in \cite{fawzi2022lower}). However, if the noise acts independently and identically on each subsystem in the memory, we can obtain a much stronger convergence estimate, as is shown below.

\subsection{IID noise}
Suppose that the memory $A$ is comprised of $m$ identical subsytems $B$ (e.g. qubits), so that the Hilbert space $\Hil_A$ factors as $\Hil_A = \Hil_B ^{\otimes m}$, and that the noise acts independently and identically, i.e., $\Phi:\cL (\cH_A)\to \cL (\cH_A)$ is of the form $\Phi = \Psi^{\otimes m}$ for some channel $\Psi:\cL (\Hil_B) \to \cL (\Hil_B)$. As the peripheral space is multiplicative under tensor product $\chi(\Phi) = \chi (\Psi)^{\otimes m}$, we get that the infinite-time capacities are additive under tensor product (see \cite[Section 3]{fawzi2024error}): 
\begin{align*}
  \forall m \in \mathbb{N}: \quad \log \max_{k} d_k \leq \lim_{n \to \infty} \frac{1}{m} Q_{\epsilon}((\Psi^{\otimes m})^n) \leq \log \max_{k} d_k + \frac{1}{m}\log(\frac{1}{1-\epsilon}),
\end{align*}
where the integers $d_k$ come from the block decomposition of $\chi (\Psi)$. Similar
results hold for all other the capacities as well. Moreover, the asymptotic part of $\Phi$ also factors like $\Phi_{\infty}=\Psi_{\infty}^{\otimes m}$, so that
\begin{align}
    \norm{\Phi^n - \Phi^n_{\infty}}_{\diamond} &= \norm{(\Psi^{\otimes m})^n - (\Psi_{\infty}^{\otimes m})^n}_{\diamond} \nonumber \\ 
    &= \norm{(\Psi^n)^{\otimes m} - (\Psi_{\infty}^n)^{\otimes m}}_{\diamond} \nonumber \\
    &\leq m \norm{\Psi^n - \Psi_{\infty}^n}_{\diamond} \leq m \kappa \mu^n,
\end{align}
where $\kappa, \mu=\operatorname{spr}(\Psi-\Psi_{\infty})$ govern the asymptotic behaviour of $\Psi$, and are hence independent of $m$. From \cite[Corollary 4.4]{Szehr2014specconvergence}, $\kappa \leq M n^{d_B^2}$, where $M$ is a constant that depends on $
\mu$ and $d_B$. Thus, in this setting, the noise reaches its infinite-time capacity when $n$ is large enough so that $mn^{d_B^2}\mu^n\leq \delta$ for a given threshold $\delta<1$, which happens when 
$$n \geq C(\log m + d_B^2 \log d_B +\log 1/\delta)/ (\log 1/\mu).$$
Note that the convergence here is incredibly rapid: for a memory comprised of $m$ qubits undergoing IID noise, the infinite time capacity is reached after time $O(\log m)$, where as in the general non-IID case, it takes time $O(2^{2m})$ to do so.


\medskip

\subsection*{Proof of Theorem~\ref{theorem:main}}  We start by proving the achievability bounds in Eqs.~\eqref{qlo}-\eqref{cealo}.
\medskip

\noindent
{\bf{Achievability: Quantum communication~\eqref{qlo}}}
\smallskip

Note that the action of a channel $\Phi$ on its peripheral space $\chi (\Phi)$ is reversible \cite[Theorem 6.16]{Wolf2012Qtour}, i.e., there exists a channel $\mathcal{R}:\B{\Hil_A}\to \B{\Hil_A}$ such that $\mathcal{R}\circ \Phi = \mathcal{P}$ (which implies that $\mathcal{R}^n\circ \Phi^n = \mathcal{P}$), where $\mathcal{P}$ is the projection onto the peripheral space (Eq.~\eqref{eq:phiinf-proj}). Thus, in the language of \cite{kribs2006error}, all the $\Hil_{k,1}$ sectors in the decomposition in Eq.~\eqref{eq:phasespace} are correctable for $\Phi^n$ for all $n\in \mathbb{N}$. Corresponding subspaces $\mathcal{C}_k \subseteq \Hil_A$ with $\dim \mathcal{C}_k = \dim{\Hil_{k,1}} = d_k$ can be constructed using \cite[Theorem 3.7]{kribs2006error} that satisfy the so-called Knill-Laflamme error-correction conditions for $\Phi^n$ for all $n\in \mathbb{N}$ \cite{knil-laf}, i.e. for all $n\in \mathbb{N}$ and $k$, $\exists$ channels $\mathcal{R}_{n,k}:\B{\Hil_A}\to \B{\Hil_A}$ such that 
    \begin{equation}
       \forall \rho\in \State{\mathcal{C}_k}: \qquad (\mathcal{R}_{n,k}\circ \Phi^n) (V_k\rho V_k^{\dagger}) = V_k\rho V_k^{\dagger} ,
    \end{equation}
    where $V_k:\mathcal{C}_k \hookrightarrow \Hil_A$ is the canonical embedding of $\mathcal{C}_k$ into $\Hil_A$.
    Thus, by choosing the encoding $\mathcal{E}_k:\B{\mathcal{C}_k} \to \B{\Hil_A}$ and decoding $\mathcal{D}_{n,k}:\B{\Hil_A}\to \B{\mathcal{C}_k}$ as follows: 
    \begin{equation}
        \mathcal{E}_k(\cdot)=V_k(\cdot)V_k^{\dagger}, \quad\text{and}\quad \mathcal{D}_{n,k}(\cdot)= V_k^{\dagger} \mathcal{R}_{n,k}(\cdot) V_k + \Tr [(\iden - V_kV^{\dagger}_k) \mathcal{R}_{n,k}(\cdot) ]\sigma_k,
    \end{equation}
    where $\sigma_k\in \State{\mathcal{C}_k}$ is some state, 
    we see that $\mathcal{D}_{n,k}\circ \Phi^n \circ \mathcal{E}_{k} = \id_{\mathcal{C}_k}$, so that $(\mathcal{E}_k,\mathcal{D}_{n,k})$ forms a $(d_k,\epsilon)$ quantum communication protocol for $\Phi^n$ with $\epsilon=0$. Hence,
    \begin{equation}
     \forall n\in \mathbb{N}, \epsilon\in [0,1): \quad   \log \max_k d_k \leq Q_0 (\Phi^n) \leq Q_{\epsilon}(\Phi^n). 
    \end{equation}
\medskip

\noindent
{\bf{Achievability: Private classical communication~\eqref{cplo}}}
\smallskip

For private classical communication, note that $Q_0 (\Phi)\leq C^{\operatorname{p}}_0 (\Phi)$ holds for any channel $\Phi$ (see Lemma~\ref{lemma:Q<P}), so that
    \begin{equation}
    \forall n\in \mathbb{N}, \epsilon\in [0,1): \quad   \log \max_k d_k \leq Q_0 (\Phi^n) \leq C^{\operatorname{p}}_0 (\Phi^n) \leq  C^{\operatorname{p}}_{\epsilon}(\Phi^n). 
    \end{equation}
  \medskip

\noindent
{\bf{Achievability: Classical communication~\eqref{clo}}}
\smallskip

    For classical communication, we can send $\sum_{k=1}^K d_k$ messages perfectly (i.e. with $\epsilon=0$ error) through $\Phi^n$ for all $n$ by using the encoding states $\{ \ketbra{i_k} \otimes \delta_k \}$ for $k=1,2,\ldots ,K$ and $i_k = 1,2,\ldots , d_k$, where $\ketbra{i_k}$ are the diagonal matrix units in $\B{\Hil_{k,1}}$ and $\delta_k$ are given in Eq.~\eqref{eq:phasespace}. Note that for each $k$, the state $\ketbra{i_k}\otimes \delta_k$ is supported only on $\Hil_{k,1}\otimes \Hil_{k,2}$. From the action of $\Phi$ on its peripheral space (see Eq.~\eqref{eq:phaseaction}), it is clear that the outputs of these states under $\Phi^n$ are mutually orthogonal for all $n$ and hence, are perfectly distinguishable. Hence, 
    \begin{equation}
    \forall n\in \mathbb{N}, \epsilon\in [0,1): \quad \log \left(\sum_{k=1}^K d_k\right)\leq C_0(\Phi^n)\leq  C_{\epsilon} (\Phi^n).
    \end{equation}
  \medskip

\noindent
{\bf{Achievability: Entanglement-assisted classical communication~\eqref{cealo}}}
 \smallskip

    With entanglement assistance, we can perfectly transmit $\sum_{k=1}^K d^2_k$ classical messages through $\Phi^n$ for all $n$. To see this, we start with an entangled state 
    \begin{equation}
        \psi_{A'A} = \frac{1}{K} \oplus_{k=1}^K (\psi^+_k \otimes \delta_k) \in \State{\Hil_{A'}\otimes \Hil_A} = \State{\oplus_k (\Hil_{A'} \otimes \Hil_{k,1}\otimes \Hil_{k,2}) },
    \end{equation}
    where $\psi^+_k\in \State{\Hil_{A'}\otimes \Hil_{k,1}}$ are maximally entangled states of Schmidt rank $d_k$, where $d_k = {\rm{dim}} \, \Hil_{k,1}$, and $\delta_k$ are the positive definite states given in Eq.~\eqref{eq:phasespace}. For each $k$, we apply an orthogonal set of unitary operators in $\B{\Hil_{k,1}}$ locally on $\Hil_{k,1}$ to encode $d_k^2$ many classical messages in orthogonal states\footnote{This is exactly the encoding scheme employed in the superdense coding protocol \cite{Bennett1992dense, Werner2001dense}.}, thus encoding $\sum_k d_k^2$ messages in total. The permutation+unitary action of $\Phi$ on its peripheral space (Eq.~\eqref{eq:phaseaction}) ensures that these states remain orthogonal (and hence perfectly distinguishable) after the action of $\Phi^n_{A\to A}$ for all $n$. Thus,
    \begin{equation}
    \forall n\in \mathbb{N}, \epsilon\in [0,1): \quad \log \left(\sum_{k=1}^K d^2_k\right)\leq  C_0^{\operatorname{ea}}(\Phi^n)\leq  C_{\epsilon}^{\operatorname{ea}} (\Phi^n).
    \end{equation}
\medskip

Next, we prove the converse bounds in Eqs.~\eqref{eq:Qconverse}-\eqref{eq:Ceaconverse}.
\smallskip

The proofs of these bounds for the quantum, classical, and entanglement-assisted classical capacities start similarly and so we consider them together below. Let us fix $\epsilon\in [0,1)$. Note that $\norm{\Phi^n-\Phi^n_{\infty}}_{\diamond} \to 0$ as $n\to \infty$ and the convergence behaves like $\norm{\Phi^n-\Phi^n_{\infty}}_{\diamond}\leq \kappa\mu^n$, so that for $n$ large enough such that $\epsilon+\kappa\mu^n < 1$, we can use Lemmas~\ref{lemma:epsilon-delta} and  \ref{lemma:bottleneck}, and the fact that $\Phi^n_{\infty}=\Phi^n_{\infty} \circ \mathcal{P} = \mathcal{P}\circ \Phi^n_{\infty}$ for all $n$ to write 
\begin{align}
    Q_{\epsilon}(\Phi^n) &\leq Q_{\epsilon + \kappa\mu^n} (\Phi^n_{\infty}) \leq Q_{\epsilon + \kappa\mu^n} (\mathcal{P}),  \\
    C_{\epsilon}(\Phi^n) &\leq C_{\epsilon + \kappa\mu^n} (\Phi^n_{\infty}) \leq C_{\epsilon + \kappa\mu^n} (\mathcal{P}), \\
    C_{\epsilon}^{\operatorname{ea}}(\Phi^n) &\leq C^{\operatorname{ea}}_{\epsilon + \kappa\mu^n} (\Phi^n_{\infty}) \leq C^{\operatorname{ea}}_{\epsilon + \kappa\mu^n} (\mathcal{P}).
\end{align}
Recall that $\mathcal{P}:\B{\Hil_A}\to \B{\Hil_A}$ is the channel that projects onto the peripheral space (Eq.~\eqref{eq:phasespace})
\begin{equation}
\chi (\Phi) = 0 \oplus \bigoplus_{k=1}^K (\B{\Hil_{k,1}}\otimes \delta_k),
\end{equation}
where the direct sum is with respect to the decomposition $\Hil_A = \Hil_{0} \oplus \bigoplus_{k=1}^K \Hil_{k,1}\otimes \Hil_{k,2}$. For our purposes, we can assume that $\Hil_0$ is the zero subspace, so that the action of $\mathcal{P}$ becomes
\begin{align}\label{eq:phase-proj}
  \forall X\in \B{\Hil_A}: \quad \mathcal{P}(X) &= \bigoplus_{k=1}^K \Tr_2 (P_k X P_k) \otimes \delta_k, 
\end{align}
where $P_k\in \B{\Hil_A}$ is the orthogonal projection that projects onto the block $\Hil_{k,1}\otimes \Hil_{k,2}$ and $\Tr_2$ denotes the partial trace over $\Hil_{k,2}$. We refer the readers to Appendix~\ref{appen:Hil0} for justification of this assumption.
\medskip

\noindent
{\bf{Converse: Quantum communication~\eqref{eq:Qconverse}}}
\smallskip

For the quantum capacity, Lemma~\ref{lemma:1shot-converse} shows that 
\begin{align}
    Q_{\epsilon + \kappa\mu^n}(\mathcal{P}) \leq \sup_{\psi_{RA}} \inf_{\sigma_A} D_{\max}(\mathcal{P}_{A\to A}(\psi_{RA}) || \iden_R \otimes \sigma_A ) + \log(\frac{1}{1-\epsilon - \kappa\mu^n}). 
\end{align}
We bound the first term above as follows. For a pure state $\psi_{RA}$, we use Eq.~\eqref{eq:phase-proj} to write
\begin{equation}
    \mathcal{P}_{A\to A}(\psi_{RA}) = \bigoplus_k \lambda_k \frac{1}{\lambda_k} \Tr_{2} \left[ (\iden_R\otimes P_k) \psi (\iden_R \otimes P_k) \right] \otimes \delta_k = \bigoplus_k \lambda_k \theta_k \otimes \delta_k, 
\end{equation}
where $\lambda_k = \Tr \left[ (\iden_R\otimes P_k) \psi (\iden_R \otimes P_k) \right]$ and each $\theta_k$ is a state in $\State{\Hil_R\otimes \Hil_{k,1}}$. Thus, by choosing $\sigma_A = \oplus_k \lambda_k \sigma_k \otimes \delta_k$, where $\sigma_k$ are arbitrary states in $\State{\Hil_{k,1}}$, we get 
\begin{align}
    \inf_{\sigma} D_{\max}(\mathcal{P}_{A\to A}(\psi_{RA}) || \iden_R \otimes \sigma_A ) &\leq \inf_{ \{ \sigma_k \}_k} 
    D_{\max}\left(\bigoplus_k \lambda_k \theta_k \otimes \delta_k \bigg\| \bigoplus_k \lambda_k \iden_R \otimes \sigma_k \otimes \delta_k \right) \nonumber \\ 
    &= \inf_{ \{\sigma_k\}_k } \max_k D_{\max} (\theta_k || \iden_R \otimes \sigma_k) \nonumber \\ 
    &= \max_k \inf_{\sigma_k} D_{\max} (\theta_k || \iden_R \otimes \sigma_k) \nonumber \\
    &\leq \log \max_k d_k,
\end{align}
where the first equality follows from quasi-convexity of $D_{\max}$ (Eq.~\eqref{eq:qconvex}), the second equality follows from Lemma~\ref{lemma:infmax}, and the last inequality follows from the fact that for any state $\rho_{AB}$,
\begin{equation}
    \inf_{\sigma\in \State{\Hil_B}} D_{\max}(\rho_{AB} \Vert \iden_A \otimes \sigma_B) \leq D_{\max}(\rho_{AB} \Vert \iden_A \otimes \iden_B/d_B) \leq \log d_B.
\end{equation}

\medskip

\noindent
{\bf{Converse: Classical communication~\eqref{eq:Cconverse}}}
\smallskip

For the classical capacity, Lemma~\ref{lemma:1shot-converse} shows that
\begin{align}
    C_{\epsilon + \kappa\mu^n}(\mathcal{P}) \leq \sup_{\rho_{MA}} \inf_{\sigma_A} D_{\max}(\mathcal{P}_{A\to A}(\rho_{MA}) || \rho_M \otimes \sigma_A ) + \log(\frac{1}{1-\epsilon - \kappa\mu^n}). 
\end{align}

For any classical-quantum (c-q) state $\rho_{MA}=\sum_m p(m)\ketbra{m}\otimes \rho^m_A$, we can write 
\begin{align}
    \inf_{\sigma}D_{\max}(\mathcal{P}_{A\to A}(\rho_{MA}) || \rho_M \otimes \sigma_A ) = \inf_{\sigma} \max_m  D_{\max}(\mathcal{P}_{A \to A}(\rho^m_A) || \sigma_A ) \leq \log (\sum_k d_k),
\end{align}
where the inequality follows by choosing $\sigma_A = (\oplus_k \iden_k \otimes \delta_k)/\sum_k d_k$ and noting that for any state $\rho$, its projection $\mathcal{P}(\rho)$ onto the peripheral space is dominated by $\oplus_k (\iden_k \otimes \delta_k )$ (see Eq.~\eqref{eq:phase-proj}).
\medskip
\medskip

\noindent
{\bf{Converse: Entanglement-assisted classical communication~\eqref{eq:Ceaconverse}}}
\smallskip

For the entanglement-assisted classical capacity, we again use Lemma~\ref{lemma:1shot-converse} to write
\begin{equation}
    C^{\operatorname{ea}}_{\epsilon+\kappa\mu^n}(\mathcal{P}) \leq \sup_{\psi_{RA}} \inf_{\sigma_A} D_{\max}(\mathcal{P}_{A\to A}(\psi_{RA}) || \psi_R \otimes \sigma_A) + \log(\frac{1}{1-\epsilon-\kappa\mu^n}), 
\end{equation}
where $d=d_A=d_R$. We note that the supremum in the above sum is achieved by the maximally entangled state $\psi^+_{RA}=1/d \sum_{i,j} \ketbra{i}{j} \otimes \ketbra{i}{j} =: \Gamma_{RA}/d$, where $\Gamma_{RA}$ is the unnormalized maximally entangled state (see \cite[Remark 2]{Fang2020smooth}). Now, by using Eq.~\eqref{eq:phase-proj}, we write
\begin{align}
     \mathcal{P}_{A\to A}(\psi^+_{RA}) = \frac{1}{d} \bigoplus_k \Tr_{2} \left[ (\iden_R\otimes P_k) \Gamma_{RA} (\iden_R \otimes P_k) \right] \otimes \delta_k = \frac{1}{d} \bigoplus_k \theta_k \otimes \delta_k, 
\end{align}
where $\theta_k = \Tr_{2} \left[ (\iden_R\otimes P_k) \Gamma_{RA} (\iden_R \otimes P_k) \right]$ is a positive operator in $\B{\Hil_R\otimes \Hil_{k,1}}$ with $\Tr \theta_k = d_k$. Let us choose $\sigma_A = \oplus_k \lambda_k (\iden_k /d_k \otimes \delta_k )$, where $\{ \lambda_k = d^2_k / \sum_{k'} d^2_{k'} \}_k$ is a probability distribution. Then,
\begin{align}
 \inf_{\sigma} D_{\max}(\mathcal{P}_{A\to A}(\psi^+_{RA}) || \psi_R^+ \otimes \sigma_A) &\leq D_{\max}\left(\frac{1}{d} \bigoplus_k \theta_k \otimes \delta_k \bigg\|  \frac{1}{d} \bigoplus_k \lambda_k \iden_R \otimes \frac{\iden_k}{d_k} \otimes \delta_k \right) \nonumber \\
  &= \max_k D_{\max} \left( \theta_k \bigg\| \lambda_k \iden_R \otimes \frac{\iden_k}{d_k} \right) \nonumber \\
  &= \log \left( \frac{d_k}{\lambda_k} \norm{\theta_k}_{\infty}  \right) \leq \log (d^2_k / \lambda_k ) = \log \left(\sum_k d^2_k\right),
\end{align}
which proves the desired result. Above, the first equality follows from quasi-convexity of $D_{\max}$ (Eq.~\eqref{eq:qconvex}) and the second equality follows from Eq.~\eqref{eq:Dmax}.
\medskip

\noindent
{\bf{Converse: Private classical communication~\eqref{eq:Pconverse}}}
\smallskip

The proof of the converse bound for private classical capacity requires a slightly different line of argumentation. We again fix $\epsilon\in [0,1)$ and note that $\norm{\Phi^n-\Phi^n_{\infty}}_{\diamond} \to 0$ as $n\to \infty$, where the convergence behaves like $\norm{\Phi^n-\Phi^n_{\infty}}_{\diamond}\leq \kappa\mu^n$, so that for $n$ large enough such that $\epsilon+\kappa\mu^n < 1$, we can use Lemmas~\ref{lemma:P1shot-converse}, \ref{lemma:Eepsilon-delta} and \ref{lemma:Ebottleneck}, and the fact that $\Phi^n_{\infty}=\Phi^n_{\infty} \circ \mathcal{P} = \mathcal{P}\circ \Phi^n_{\infty}$ for all $n$ to write 
\begin{align}
    C^{\operatorname{p}}_{\epsilon}(\Phi^n) \leq E_H^{\epsilon} (\Phi^n) &\leq E_H^{\epsilon + \kappa\mu^n} (\Phi^n_{\infty})  \nonumber\\ 
    &\leq E_H^{\epsilon + \kappa\mu^n} (\mathcal{P}) \nonumber \\
    &\leq \sup_{\psi_{RA}} \inf_{\sigma_{RA}\in \operatorname{SEP}(R:A)} D_{\max} (\mathcal{P}_{A\to A}(\psi_{RA}) \Vert \sigma_{RA}) + \log (\frac{1}{1-\epsilon-\kappa\mu^n}).
\end{align}
We bound the first term above as follows. For a pure state $\psi_{RA}$, we use Eq.~\eqref{eq:phase-proj} to write
\begin{equation}
    \mathcal{P}_{A\to A}(\psi_{RA}) = \bigoplus_k \lambda_k \frac{1}{\lambda_k} \Tr_{2} \left[ (\iden_R\otimes P_k) \psi (\iden_R \otimes P_k) \right] \otimes \delta_k = \bigoplus_k \lambda_k \theta_k \otimes \delta_k, 
\end{equation}
where $\lambda_k = \Tr \left[ (\iden_R\otimes P_k) \psi (\iden_R \otimes P_k) \right]$ and each $\theta_k$ is a state in $\State{\Hil_R\otimes \Hil_{k,1}}$. Thus, by choosing $\sigma_{RA} = \oplus_k \lambda_k \sigma_k \otimes \delta_k$, where $\sigma_k$ are arbitrary separable states in $\State{\Hil_R\otimes \Hil_{k,1}}$, we get 
\begin{align}
    \inf_{\sigma\in \operatorname{SEP}(R:A)} D_{\max}(\mathcal{P}_{A\to A}(\psi_{RA}) || \sigma_{RA} ) &\leq \inf_{ \{ \sigma_k \}_k } 
    D_{\max}\left(\bigoplus_k \lambda_k \theta_k \otimes \delta_k \bigg\| \bigoplus_k \lambda_k \sigma_k \otimes \delta_k \right) \nonumber \\ 
    &= \inf_{ \{\sigma_k\}_k } \max_k D_{\max} (\theta_k || \sigma_k) \nonumber \\ 
    &= \max_k \inf_{\sigma_k} D_{\max} (\theta_k || \sigma_k) \nonumber \\
    &\leq \log ( \max_k d_k),
\end{align}
where the first equality follows from Eq.~\eqref{eq:qconvex}, the second equality follows from Lemma~\ref{lemma:infmax}, and the last inequality follows from the fact that for any state $\rho_{AB}$ (see Lemma~\ref{lemma:Emax}),
\begin{equation}
    \inf_{\sigma\in \operatorname{SEP}(A:B)} D_{\max}(\rho_{AB}\Vert \sigma_{AB}) \leq \log \min \{d_A,d_B \}.
\end{equation}
$\qed$

\begin{remark}
    The achievability bounds in Eqs.~\eqref{qlo}-\eqref{cealo} of Theorem~\ref{theorem:main} are obtained by constructing communication protocols that work with $\epsilon=0$ error. It is unclear whether these bounds can be improved by explicitly taking $\epsilon$ into account. In this regard, we note that the achievability bounds in \cite{fawzi2024error} on the quantum and classical capacities do take $\epsilon$ into account and are slightly better than the ones in Eqs.~\eqref{qlo},\eqref{clo}. However, the error criteria they consider when defining the capacities are of the average kind, as opposed to the worst case error criteria that we employ. 
\end{remark}

\section{Conclusion}
Our work provides a starting point to understand the fundamental limits on error-correction in a dynamical framework. By eliminating all computational and physical constraints on the encoding and decoding operations, we have presented an information-theoretic analysis of the most general information storage capacities of a quantum memory device. Two critical assumptions of our model are worth highlighting:
\begin{itemize}
    \item \textbf{Markovian noise:} We assume that the memory experiences Markovian noise which can be modelled as a dQMS. The physical justification for this assumption is debatable. However, it does provide a good starting point for our study, and is the standard assumption in the study of quantum memories coupled to heat baths at finite temperatures (see \cite{Brown2016memory}). Note that in this regard, it is common to impose locality constraints on the interaction between the memory and bath, so that at each time step, the thermal interaction affects only a fixed number of physical qubits in memory. Our model, on the other hand, has no such restriction, and we allow the interaction to arbitrarily affect all qubits in the memory at each time step.
    \item \textbf{Passive error correction:} We assume a passive model for error-correction, i.e., we do not allow error-correction to occur in between time steps. However, note that our model can cover a fixed error-correction mechanism at each time step by setting $\Phi=\Phi_{\operatorname{ecc}}\circ \Phi_{\operatorname{noise}}$, where $\Phi_{\operatorname{ecc}}$ is the error-correction mechanism and $\Phi_{\operatorname{noise}}$ is the noise.
\end{itemize}

It would be interesting to see to what extent do the current results hold when the stated assumptions are relaxed. Furthermore, in the current model, one can ask for what noise models $\Phi$ do the capacities converge to their infinite-time values faster than the rate proposed in Section~\ref{subsec:convergence}.

\section{Acknowledgements}
We thank Omar Fawzi for helpful comments on the first version of this manuscript.

\appendix

\section{Technical results} \label{appen:tech}
\begin{lemma}\label{lemma:Q<P}
    For a quantum channel $\Phi:\B{\Hil_A}\to \B{\Hil_B}$, the zero-error one-shot quantum and private capacities satisfy the relation $Q_0(\Phi)\leq C^{\operatorname{p}}_0(\Phi)$.
\end{lemma}
\begin{proof}
    Consider a $(\mathscr{M}, 0)$ quantum communication protocol $(\mathcal{E}_{A'\to A}, \mathcal{D}_{B\to B'})$ for $\Phi$ with $\mathscr{M}=d_{A'}=d_{B'}=d_R$, which we can use to transmit one-half of a maximally state 
    \begin{equation}
        \psi^+_{RA'} = \frac{1}{\mathscr{M}} \sum_{1\leq m,m'\leq \mathscr{M}} \ketbra{m}{m'}_R \otimes \ketbra{m}{m'}_{A'}
    \end{equation}
    of Schmidt rank $\mathscr{M}$ through $\Phi$ perfectly, i.e.
    \begin{equation}
        \psi^+_{RB'} = \mathcal{D}_{B\to B'}\circ \Phi_{A\to B}\circ \mathcal{E}_{A'\to A}(\psi^+_{RA'}).
    \end{equation}
    Let $\mathcal{V}_{A\to BE}$ be an isometric extension of $\Phi_{A\to B}$ and consider the state 
    \begin{equation}
        \omega_{RB'E} = \mathcal{D}_{B\to B'}\circ \mathcal{V}_{A\to BE}\circ \mathcal{E}_{A'\to A}(\psi^+_{RA'}),
    \end{equation}
    which extends the state at the output of the protocol, i.e.,  $\omega_{RB'} = \mathcal{D}_{B\to B'}\circ \Phi_{A\to B}\circ \mathcal{E}_{A'\to A}(\psi^+_{RA'})$. Since the only possible extension of $\psi^+_{RB'}$ is of the form $\psi^+_{RB'}\otimes \sigma_E$ for some state $\sigma_E$, we get
    \begin{equation}
        \omega_{RB'E} = \psi^+_{RB'}\otimes \sigma_E =  \mathcal{D}_{B\to B'}\circ \mathcal{V}_{A\to BE}\circ \mathcal{E}_{A'\to A}(\psi^+_{RA'}).
    \end{equation}
    Applying a measurement with POVMs $\{ \ketbra{m}_R \}_{m\in [\mathscr{M}]}$ and $\{ \ketbra{m}_{B'} \}_{m\in [\mathscr{M}]}$ on the $R$ and $B'$ systems yields    \begin{equation}\label{eq:star}
        \overbar{\psi}_{RB'}^+ \otimes \sigma_E =  \overbar{\mathcal{D}}_{B\to B'}\circ \mathcal{V}_{A\to BE}\circ \mathcal{E}_{A'\to A}(\overbar{\psi}^+_{RA'}),
    \end{equation}
    where $\overbar{\psi}_{RA'}^+ = \frac{1}{\mathscr{M}} \sum_{1\leq m\leq \mathscr{M}} \ketbra{m}{m}_R \otimes \ketbra{m}{m}_{A'}$ is a maximally classically correlated state and $\overbar{\mathcal{D}}_{B\to B'}$ is a measurement channel defined as 
    \begin{align}
        \overbar{\mathcal{D}}_{B\to B'}(X_B) &= \sum_m \Tr (\ketbra{m}_{B'} \mathcal{D}_{B\to B'}(X_B) ) \ketbra{m}_{B'} \nonumber \\ 
        &= \sum_m \Tr ( \mathcal{D}_{B'\to B}^*( \ketbra{m}_{B'}) X_B ) \ketbra{m}_{B'}. 
    \end{align} 
    Note that $\{ \mathcal{D}_{B'\to B}^*( \ketbra{m}_{B'} \}_{m=1}^{\mathscr{M}}$ forms a POVM, since $\mathcal{D}_{B'\to B}^*$ is unital. Expanding the LHS and RHS of the Eq.~\eqref{eq:star} by using the formula for $\overbar{\psi}^+$ and matching terms shows that for each $m\in [\mathscr{M}]$,
    \begin{equation}
        \ketbra{m}_{B'}\otimes \sigma_E = \overbar{\mathcal{D}}_{B\to B'}\circ \mathcal{V}_{A\to BE} (\rho^m_{A}),
    \end{equation}
    where the states $\rho^m$ are defined as $\rho^m_A = \mathcal{E}_{A'\to A}(\ketbra{m}_{A'})$. Thus, the encoding states $\rho^m_A$ for $m\in [\mathscr{M}]$ and the decoding POVM $\{ \mathcal{D}_{B'\to B}^*( \ketbra{m}_{B'} \}_{m=1}^{\mathscr{M}}$ forms a $(\mathscr{M},0)$ private classical communication protocol for $\Phi$. Since the quantum communication protocol that we started with was arbitrary, we obtain the desired result.
\end{proof}

\begin{lemma}\label{lemma:1shot-converse}
    For any channel $\Phi:\B{\Hil_A}\to \B{\Hil_B}$ and $\epsilon\in [0,1)$, the following bounds hold:
    \begin{align}
        Q_{\epsilon}(\Phi) &\leq \sup_{\psi_{RA}} \inf_{\sigma_B} D_{\max}(\Phi_{A\to B}(\psi_{RA}) || \iden_R \otimes \sigma_B) + \log(\frac{1}{1-\epsilon}), \\
        C_{\epsilon}(\Phi) &\leq \sup_{\rho_{MA}} \inf_{\sigma_B} D_{\max}(\Phi_{A\to B}(\rho_{MA}) || \rho_M \otimes \sigma_B) + \log(\frac{1}{1-\epsilon}),  \\ 
        C^{\operatorname{ea}}_{\epsilon}(\Phi) &\leq \sup_{\psi_{RA}} \inf_{\sigma_B} D_{\max}(\Phi_{A\to B}(\psi_{RA}) || \psi_R \otimes \sigma_B) + \log(\frac{1}{1-\epsilon}), 
    \end{align}
    where the supremum is either over pure states $\psi_{RA}\in \State{\Hil_R\otimes \Hil_A}$ with $d_R=d_A$ or classical-quantum states $\rho_{MA}=\sum_m p(m) \ketbra{m}_M\otimes \rho^m_A\in \State{\Hil_M\otimes \Hil_A}$ , and the infimum is over arbitrary states $\sigma_B\in \State{\Hil_B}$.
\end{lemma}
\begin{proof}
    For a unified account of these bounds, we refer the readers to \cite[Chapters 11-14]{khatri2024principles}. In particular, we make use of
    \cite[Theorem 11.6, Theorem 12.4, Corollary 14.4]{khatri2024principles}. See also the Bibliographic notes in \cite{khatri2024principles} for references to original papers where these bounds were first established.
\end{proof}

\begin{lemma}\label{lemma:P1shot-converse}
    For any channel $\Phi:\B{\Hil_A}\to \B{\Hil_B}$ and $\epsilon\in [0,1)$, the following bound holds:
    \begin{equation}
        C^{\operatorname{p}}_{\epsilon} (\Phi) \leq E_H^{\epsilon}(\Phi) \leq \sup_{\psi_{RA}} \inf_{\sigma_{RB}\in \operatorname{SEP}(R:B)} D_{\max} (\Phi_{A\to B}(\psi_{RA}) \Vert \sigma_{RB}) + \log (\frac{1}{1-\epsilon}),
    \end{equation}
    where the supremum if over all pure states $\psi_{RA}$ with $d_R=d_A$.
\end{lemma}
\begin{proof}
    The first inequality was proved in \cite[Theorem 11]{Wilde2017private}. The second inequality follows from the inequality
    \begin{equation}
        D_H^{\epsilon}(\rho \Vert \sigma) \leq D_{\max} (\rho \Vert \sigma) + \log (\frac{1}{1-\epsilon}),
    \end{equation}
    which holds for any two states $\rho,\sigma \in \State{\Hil_A}$ and was established in \cite[Lemma 5]{Cooney2016hypo}. 
\end{proof}

\begin{lemma}\label{lemma:epsilon-delta}
    Let $\Phi, \Psi : \B{\Hil_A}\to \B{\Hil_B}$ be quantum channels such that $\norm{\Phi-\Psi}_{\diamond}\leq \delta$. Then, for $\epsilon\in [0,1)$ such that $\epsilon+\delta<1$:
    \begin{align}
        Q_{\epsilon}(\Phi) &\leq Q_{\epsilon+\delta}(\Psi) \\
        C_{\epsilon}(\Phi) &\leq C_{\epsilon+\delta}(\Psi) \\
        C^{\operatorname{ea}}_{\epsilon}(\Phi) &\leq C^{\operatorname{ea}}_{\epsilon+\delta}(\Psi)
    \end{align}
\end{lemma}
\begin{proof}
Consider a $(M,\epsilon)$ classical communication protocol $\{\rho^m_A, \Lambda^m_B \}_{m=1}^M$ for $\Phi_{A\to B}$ as in Section~\ref{subsec:cc}. This means that for each message $m$,
\begin{equation}
    \Tr [\Lambda^m_{B}  (\Phi_{A\to B}(\rho^m_A)  ] \geq 1 -\epsilon.
\end{equation}
For each $m$, it is then easy to see that
\begin{align}
    \Tr [\Lambda^m_B  (\Psi_{A\to B}(\rho^m_A)  ] &= \Tr [\Lambda^m_B  (\Phi_{A\to B}(\rho^m)  ] - \Tr [\Lambda^m_B  (\Phi - \Psi)_{A\to B}(\rho^m_A)  ]  \nonumber \\
    &\geq 1- (\epsilon + \delta),
\end{align}
where the last inequality follows from the fact that 
\begin{align}
    \Tr [\Lambda^m_B  (\Phi - \Psi)_{A\to B}(\rho^m_A)  ] &\leq \norm{\Lambda^m_B}_{\infty} \norm{(\Phi-\Psi)_{A\to B}(\rho^m_A)}_1 \nonumber \\
    &\leq \norm{\Phi-\Psi}_{\diamond} \leq \delta.
\end{align}
Hence, $\{\rho^m_A, \Lambda^m_B \}_{m=1}^M$ works as a $(M,\epsilon+\delta)$ classical communication protocol for $\Psi_{A\to B}$.

Consider a $(d,\epsilon)$ quantum communication protocol $(\mathcal{E}_{A'\to A}, \mathcal{D}_{B\to B'})$ for $\Phi_{A\to B}$ such that 
\begin{equation}
\forall \psi_{RB'}: \quad  \Tr[\psi_{RB'} (\mathcal{D}_{B\to B'}\circ \Phi_{A\to B}\circ \mathcal{E}_{A'\to A}(\psi_{RA'})) ] \geq  1-\epsilon,
\end{equation}
where $d=d_{A'}=d_{B'}=d_R$. For any $\psi_{RA'}$, it is then easy to see that
\begin{align}
    &\Tr[\psi_{RB'} (\mathcal{D}_{B\to B'}\circ \Psi_{A\to B}\circ \mathcal{E}_{A'\to A}(\psi_{RA'}) )] \nonumber  \\ 
    &= \Tr[\psi_{RB'} (\mathcal{D}_{B\to B'}\circ \Phi_{A\to B}\circ \mathcal{E}_{A'\to A}(\psi_{RA'})) ] - \Tr[\psi_{RB'} (\mathcal{D}_{B\to B'}\circ (\Phi-\Psi)_{A\to B}\circ \mathcal{E}_{A'\to A}(\psi_{RA'})) ] \nonumber \\
    &\geq 1-(\epsilon + \delta),
\end{align}
where the last inequality follows from the fact that
\begin{align}
    \Tr[\psi_{RB'} (\mathcal{D}_{B\to B'}\circ (\Phi-\Psi)_{A\to B}\circ \mathcal{E}_{A'\to A}(\psi_{RA'})) ] &\leq \norm{\mathcal{D}_{B\to B'}\circ (\Phi-\Psi)_{A\to B}\circ \mathcal{E}_{A'\to A}}_{\diamond} \nonumber \\
    &\leq \norm{\Phi-\Psi}_{\diamond} \leq \delta,
\end{align}
where we have used sub-multiplicativity of the diamond norm \cite[Proposition 3.48]{watrous2018theory} and the fact that $\norm{\Phi}_{\diamond}=1$ for any channel $\Phi$ \cite[Proposition 3.44]{watrous2018theory}. Thus, $(\mathcal{E}_{A'\to A}, \mathcal{D}_{B\to B'})$ works as a $(d,\epsilon+\delta)$ quantum communication protocol for $\Psi_{A\to B}$. The proof for the entanglement-assisted classical capacity works similarly.
\end{proof}

\begin{lemma}\label{lemma:EHstate}
    Let $\rho_{AB}, \omega_{AB}$ be states such that $\norm{\rho-\omega}_1 \leq \delta$. Then, $E_{H}^{\epsilon}(A:B)_{\rho} \leq E_{H}^{\epsilon+\delta}(A:B)_{\omega}$.
\end{lemma}
\begin{proof}
    Fix a separable state $\sigma_{AB}$. Then, for every $\Lambda\in \B{\Hil_A\otimes \Hil_B}$ satisfying $0\leq \Lambda \leq \iden$ and $\Tr \Lambda \rho \geq 1-\epsilon$, we have $\Tr \Lambda \omega = \Tr \Lambda \rho - \Tr \Lambda (\rho - \omega) \geq 1 - (\epsilon + \delta) $. Hence, $\beta^{\epsilon}_H (\rho \Vert \sigma) \geq \beta^{\epsilon+\delta}_H (\omega \Vert \sigma)$ and $D^{\epsilon}_H (\rho \Vert \sigma) \leq D^{\epsilon+\delta}_H (\omega \Vert \sigma)$. The claim follows by taking an infimum over all separable states $\sigma_{AB}$.
\end{proof}

\begin{lemma}\label{lemma:Eepsilon-delta}
    Let $\Phi, \Psi : \B{\Hil_A}\to \B{\Hil_B}$ be quantum channels such that $\norm{\Phi-\Psi}_{\diamond}\leq \delta$. Then,
    \begin{align}
        E_H^{\epsilon}(\Phi) \leq E_H^{\epsilon+\delta}(\Psi)
    \end{align}
\end{lemma}
\begin{proof}
    For any pure state $\psi_{RA}$, $\norm{\Phi_{A\to B}(\psi_{RA}) - \Psi_{A\to B}(\psi_{RA}) }_1 \leq \delta$ because $\norm{\Phi-\Psi}_{\diamond}\leq \delta$. Thus, the claim follows from the definition (Eq.~\eqref{eq:EHchannel}) and Lemmma~\ref{lemma:EHstate}.
\end{proof}

\begin{lemma}\label{lemma:bottleneck}
Let $\Psi_{A\to B}$, $\Phi_{B\to C}$ be quantum channels. Then, for any $\epsilon\in [0,1)$:
    \begin{align}
        Q_{\epsilon}(\Phi\circ \Psi) &\leq \min\{Q_{\epsilon}(\Phi), Q_{\epsilon}(\Psi) \}, \\
        C_{\epsilon}(\Phi\circ \Psi) &\leq \min\{C_{\epsilon}(\Phi), C_{\epsilon}(\Psi) \},  \\
        C^{\operatorname{ea}}_{\epsilon}(\Phi\circ \Psi) &\leq \min\{C^{\operatorname{ea}}_{\epsilon}(\Phi), C^{\operatorname{ea}}_{\epsilon}(\Psi) \}.  
    \end{align}
\end{lemma}
\begin{proof}
    Consider a $(d,\epsilon)$ quantum communication protocol $(\mathcal{E}_{A'\to A}, \mathcal{D}_{C\to C'})$ for $\Phi\circ \Psi$, with $d=d_{A'}=d_{C'}$, such that for any pure state $\psi_{RA'}\in \State{\Hil_R\otimes \Hil_{A'}}$
    \begin{equation}
   \bra{\psi_{RC'}} \mathcal{D}_{C\to C'}\circ (\Phi\circ\Psi)_{A\to C}\circ \mathcal{E}_{A'\to A}(\psi_{RA'})\ket{\psi_{RC'}} \geq  1-\epsilon.
\end{equation}
Now, by absorbing either $\Psi$ into the encoding channel $\mathcal{E}_{A'\to A}$ or $\Phi$ into the decoding channel $\mathcal{D}_{C\to C'}$, we see that the same $(d,\epsilon)$ protocol works for $\Phi$ and $\Psi$, which proves the desired result. We leave similar proofs for the other capacities to the reader.
\end{proof}

\begin{lemma}\label{lemma:Ebottleneck}
    Let $\Psi_{A\to B}, \Phi_{B\to C}$ be quantum channels. Then, for any $\epsilon\in [0,1)$:
    \begin{equation}
           E_H^{\epsilon}(\Phi\circ \Psi) \leq \min\{ E_H^{\epsilon}(\Phi), E_H^{\epsilon}(\Psi) \}.
    \end{equation}
\end{lemma}
\begin{proof}
    The inequality $E_H^{\epsilon}(\Phi\circ \Psi) \leq E_H^{\epsilon}(\Psi)$ is an easy consequence of the fact that the $\epsilon-$hypothesis testing relative entropy $D_H^{\epsilon}$ satisfies the data processing inequality. For the inequality $E_H^{\epsilon}(\Phi\circ \Psi) \leq E_H^{\epsilon}(\Phi)$, note that for any state $\rho_{RA}$ and $\omega_{RB}=\Psi_{A\to B}(\rho_{RA})$,
    \begin{align}
        \inf_{\sigma\in \operatorname{SEP}(R:C)} D_H^{\epsilon} (\Phi_{B\to C}(\Psi_{A\to B}(\rho_{RA})) \Vert \sigma_{RC}) = \inf_{\sigma\in \operatorname{SEP}(R:C)} D_H^{\epsilon} (\Phi_{B\to C}(\omega_{RB}) \Vert \sigma_{RC}) \leq E_H^{\epsilon}(\Phi).
    \end{align}
\end{proof}

\begin{lemma}\label{lemma:Emax}
    For any state $\rho_{AB}\in \State{\Hil_A \otimes \Hil_B}$,
\begin{equation}
    \inf_{\sigma\in \operatorname{SEP}(A:B)} D_{\max}(\rho_{AB}\Vert \sigma_{AB}) \leq \log \min \{d_A,d_B \}.
\end{equation}
\end{lemma}
\begin{proof}
    Consider the spectral decomposition $\rho_{AB}=\sum_i p_i \psi^i_{AB}$, where $\psi^i_{AB}$ are pure states and $\sum_i p_i =1$, so that we can write
    \begin{align}
        \inf_{\sigma\in \operatorname{SEP}(A:B)}D_{\max}(\rho_{AB}\Vert \sigma_{AB}) &\leq \inf_{\{\sigma^i \}_i \subset \operatorname{SEP}(A:B) } D_{\max}\left( \sum_i p_i \psi^i_{AB} \bigg\| \sum_i p_i\sigma^i_{AB} \right) \nonumber\\ 
        &\leq \inf_{\{\sigma^i \}_i \subset \operatorname{SEP}(A:B) } \max_i D_{\max} (\psi^i_{AB}\Vert \sigma^i_{AB}) \nonumber \\
        &= \max_i \inf_{\sigma^i \in \operatorname{SEP}(A:B) } D_{\max} (\psi^i_{AB}\Vert \sigma^i_{AB}).
    \end{align}
    Note that we made use of quasi-convexity of $D_{\max}$ (Eq.~\eqref{eq:qconvex}) to obtain the second inequality above and of Lemma~\ref{lemma:infmax} to obtain the last equality. Thus, it suffices to prove the claim for pure states. For a pure state $\psi_{AB}$, it is known that \cite{Datta2009maxrel}
    \begin{equation}
        \inf_{\sigma \in \operatorname{SEP}(A:B) } D_{\max} (\psi_{AB}\Vert \sigma_{AB}) = 2\log (\sum_{i=1}^d\lambda_i),
    \end{equation}
    where $\lambda_i\geq 0$ are the Schmidt coefficients of $\psi_{AB}$ satisfying $\sum_{i=1}^d \lambda_i^2=1$ and $d=\min \{d_A,d_B \} $. A simple application of Cauchy-Schwarz inequality then shows
    \begin{equation}
        \sum_{i} \lambda_i \leq \sqrt{d \sum_i \lambda_i^2} = \sqrt{d}, 
    \end{equation}
    Hence,
    \begin{equation}
        \inf_{\sigma \in \operatorname{SEP}(A:B) } D_{\max} (\psi_{AB}\Vert \sigma_{AB}) = 2\log (\sum_{i=1}^d\lambda_i) \leq 2 \log \sqrt{d} = \log d = \log \min \{d_A,d_B \}.
    \end{equation}
\end{proof}

\begin{lemma}\label{lemma:infmax}
    Let $f_k : \mathcal{S}\to \mathbb{R}$ for $k=1,2,\ldots ,K$ be arbitrary mappings, where $\mathcal{S}$ is an arbitrary set. Then, 
    \begin{equation}
        \inf_{\{x_k \}_k \subset \mathcal{S} } \left( \max_k f_k (x_k) \right) = \max_k \left(\inf_{x\in \mathcal{S}} f_k (x) \right).
    \end{equation}
\end{lemma}
\begin{proof}
    Clearly, for any subset $\{x_k \}_k \subset \mathcal{S}$, we have
    \begin{equation}
       \max_k \left(\inf_{x\in \mathcal{S}} f_k (x) \right) \leq \max_k f_k (x_k), 
    \end{equation}
    so that 
    \begin{equation}
       \max_k \left(\inf_{x\in \mathcal{S}} f_k (x) \right) \leq \inf_{\{x_k \}_k \subset \mathcal{S} } \left( \max_k f_k (x_k) \right). 
    \end{equation}
    Next we justify that the above inequality cannot be strict. Note that for any $\delta>0$, the number $\max_k \left(\inf_{x\in \mathcal{S}} f_k (x) \right) + \delta$, by definition, cannot be a lower bound on the sets $\{f_k (x) \}_{x\in \mathcal{S}}$ for all $k$. In other words, for every $\delta>0$, there exists a subset $\{x_k \}_k \subset \mathcal{S}$ such that for each $k$, $f_k(x_k) < \max_k \left(\inf_{x\in \mathcal{S}} f_k (x) \right) + \delta$, which means that  
    \begin{equation}
      \max_k f_k (x_k) <  \max_k \left(\inf_{x\in \mathcal{S}} f_k (x) \right) + \delta. 
    \end{equation}
    Hence, $\max_k \left(\inf_{x\in \mathcal{S}} f_k (x) \right)$ must be the greatest lower bound on the set $\{ \max_k f_k (x_k) \}_{\{x_k \}_k \subset \mathcal{S} }$, which is what we want to prove:
    \begin{equation}
        \max_k \left(\inf_{x\in \mathcal{S}} f_k (x) \right) = \inf_{\{x_k \}_k \subset \mathcal{S} } \left( \max_k f_k (x_k) \right). 
    \end{equation}
\end{proof}

\section{Justification for the assumption that $\Hil_0=\{0\}$} \label{appen:Hil0}

Consider the decomposition $\Hil_A = \Hil_0 \oplus \Hil_0^{\perp}$, where we have identified $\Hil_0^{\perp}=\bigoplus_{k=1}^K \Hil_{k,1}\otimes \Hil_{k,2}$. Let $V:\Hil_0^{\perp} \hookrightarrow \Hil_A $ be the canonical embedding. Then, since $\mathcal{P}:\B{\Hil_A}\to \B{\Hil_A}$ projects onto the peripheral space $\chi (\Phi) = 0 \oplus \bigoplus_{k=1}^K (\B{\Hil_{k,1}}\otimes \delta_k)$, it is clear that
\begin{equation}
    \forall X\in \B{\Hil_A}: \quad \mathcal{P}(X) = 0 \oplus V^{\dagger} \mathcal{P}(X) V = 0 \oplus \mathcal{R}_V \circ \mathcal{P} (X),
\end{equation}
where the channel $\mathcal{R}_V:\B{\Hil_A}\to \B{\Hil_0^{\perp}}$ is defined as $\mathcal{R}_V (Y) = V^{\dagger} Y V + \Tr [(\iden - VV^{\dagger})Y ]\sigma $ for some state $\sigma\in \State{\Hil_0^{\perp}}$. Moreover, since $\mathcal{P}=\mathcal{P}^2$, we get
\begin{equation}
    \mathcal{P}(X) = \mathcal{P}(\mathcal{P}(X)) = \mathcal{P}(0 \oplus \mathcal{R}_V \circ \mathcal{P} (X)) = 0 \oplus \overbar{\mathcal{P}}\circ \mathcal{R}_V\circ \mathcal{P} (X),
\end{equation}
where $\overbar{\mathcal{P}}:\B{\Hil_0^{\perp}}\to \B{\Hil_0^{\perp}}$ is defined as (see \cite[Theorem 12]{Lami2016entsaving})
\begin{align}
  \forall X\in \B{\Hil_0^{\perp}}: \quad \overbar{\mathcal{P}}(X) &= \bigoplus_{k=1}^K \Tr_2 (P_k X P_k) \otimes \delta_k. 
\end{align}
Here, $P_k\in \B{\Hil_0^{\perp}}$ is the orthogonal projection that projects onto the block $\Hil_{k,1}\otimes \Hil_{k,2}$ and $\Tr_2$ denotes the partial trace over $\Hil_{k,2}$. Therefore, Lemma~\ref{lemma:bottleneck} show that the capacities of $\mathcal{P}$ are upper bounded by those of $\overbar{\mathcal{P}}$. Hence, we can assume that $\Hil_0$ is the zero subspace, because if not, we can just work with the $\overbar{\mathcal{P}}$ channel instead of $\mathcal{P}$.

\bibliography{references}
\bibliographystyle{alpha}

\end{document}